\documentclass[12pt]{amsart}
\usepackage[margin=1in]{geometry}
\usepackage[T1]{fontenc}
\usepackage{lmodern}
\usepackage{amsmath,amsthm,amsfonts,latexsym,amssymb,amscd,color}
\usepackage{tikz}
\usetikzlibrary{chains,automata,positioning,arrows}
\usepackage{algorithm}
\usepackage{algpseudocode}
\usepackage{xurl}
\usepackage{enumitem}
\usepackage{listings}
\usepackage{mdframed}
\usepackage{bm}
\usepackage{graphicx}
\usepackage{braket}
\usepackage{hyperref}
\hypersetup{
  hidelinks,
  pdftitle={A characteristic-graph reduction for deterministic L-system inference with classical and quantum-heuristic solvers},
  pdfauthor={Ali Lotfi, Ian McQuillan, Steven Rayan},
  pdfsubject={Author-prepared manuscript incorporating published corrections},
  pdfkeywords={Deterministic L-system inference, D0L-system, characteristic graph, maximum independent set, SAT, QAOA}
}
\usepackage{cleveref}
\usepackage[linguistics]{forest}
\usepackage{float}
\usepackage{array}
\usepackage{caption}
\usepackage{placeins}
\newmdenv[linewidth=2pt,roundcorner=10pt,backgroundcolor=white,font=\bfseries]{questionbox}
\makeatletter
\providecommand{\theHALG@line}{}
\renewcommand{\theHALG@line}{\thealgorithm.\arabic{ALG@line}}
\makeatother
\newlength\myindent
\let\oldReturn\Return
\renewcommand{\Return}{\State\oldReturn}

\newtheorem{thm}{Theorem}[section]
\newtheorem{cor}[thm]{Corollary}

\newtheorem{prp}[thm]{Proposition}

\newtheorem{clm}[thm]{Claim}

\newtheorem*{clm*}{Claim}
\newtheorem*{remark}{Remark}

\newenvironment{cpf}{\begin{proof}[Proof of Claim]}{\end{proof}}

\theoremstyle{definition}

\newtheorem{exmp}[thm]{Example}

\newtheorem{remrk}[thm]{Remark}

\numberwithin{equation}{section}

\definecolor{mypink2}{RGB}{219, 48, 122}

\usepackage{booktabs}
\newcounter{constant}

\begin{document}

\title[Characteristic-graph reduction for D0L inference]{A characteristic-graph reduction for deterministic L-system inference with classical and quantum-heuristic solvers}

\author{Ali Lotfi}
\address{Department of Computer Science, University of Saskatchewan, 110 Science Place, Saskatoon, SK, CANADA, S7N 5C9}
\email{all054@mail.usask.ca}
\author{Ian McQuillan}
\address{Department of Computer Science, University of Saskatchewan, 110 Science Place, Saskatoon, SK, CANADA, S7N 5C9}
\email{mcquillan@cs.usask.ca}
\author{Steven Rayan}
\address{Centre for Quantum Topology and Its Applications (quanTA) and Department of Mathematics and Statistics, University of Saskatchewan, 106 Wiggins Road, Saskatoon, SK, CANADA, S7N 5E6}
\email{rayan@math.usask.ca}

\date{}
\thanks{Author-prepared manuscript incorporating corrections in the published article.
Published in \emph{Research in the Mathematical Sciences} \textbf{13}, article 87 (2026).
The Version of Record is available at \url{https://doi.org/10.1007/s40687-026-00667-2}.}

\begin{abstract}
L-systems are widely used to model and simulate biological growth processes, including plant development. Motivated by the broader goal of automating L-system construction, this paper studies the formal deterministic context-free L-system (D0L-system) inference problem in which the input is a finite, noise-free sequence of strings. We introduce the characteristic graph of a string sequence and prove that the existence of a compatible D0L-system is equivalent to the existence of an independent set of a prescribed size in this graph. This gives polynomial reductions from D0L inference to maximum independent set (MIS) and SAT, and hence an exact classical decision route when combined with an exact MIS or SAT solver. We also describe a Quantum Approximate Optimization Algorithm (QAOA)-based quantum heuristic for generating candidate independent sets. This heuristic is not a decision procedure for D0L inference, and failure of the heuristic is inconclusive rather than a certificate of nonexistence. We include computational validation on $22$ generated D0L traces and $8$ negative controls, explicitly incorporate known fixed-identity symbols, and report small statevector QAOA simulations. The quantum component is presented as a proof-of-concept candidate generator, not as a claimed quantum speedup; the immediately useful contribution is the characteristic-graph reduction and the exact classical/SAT decision framework it enables.
\end{abstract}

\maketitle

\tableofcontents

\section{Introduction}

\emph{Lindenmayer systems}, abbreviated as L-systems, were originally introduced in \cite{lindenmayer1968mathematical}. They are rewriting systems or formal grammar systems where the letters of a string are rewritten in parallel, and this process continues to produce a sequence of strings. One of their applications is modeling plant growth and biological systems, such as plant morphology \cite{prusinkiewicz2012algorithmic}. When the letters of an L-system are interpreted as drawing instructions, L-system simulators such as vlab \cite{algorithmicbotany}, can produce accurate graphical depictions of plant development over time.

Developing an L-system model for a growing plant currently requires expertise in L-system construction and can be time-consuming \cite{cieslak2022system}. An eventual automated pipeline from observations would also require a separate image-analysis stage: one would need to segment plant images, identify symbolic modules, order them into strings, and handle noise. Such a pipeline is an important long-term motivation, especially because L-system simulations can generate synthetic images for downstream machine-learning applications \cite{david2020global}. It is not, however, the problem solved in this paper. Here we study the exact symbolic problem: the input is already a finite sequence of strings $\theta=(w_0,w_1,\ldots,w_m)$ over a finite alphabet, and the task is to recover a D0L-system compatible with those strings or to certify that none exists. We return to this boundary, and to the related question of how many distinct D0L-systems are consistent with a given $\theta$, in Section~\ref{sec:experiments}.

While classical algorithms have been proposed for the inductive inference problem for deterministic L-systems (D0L-systems, \cite{BernardMcQuillan2021}) and stochastic L-systems (\cite{bernard2023stochastic}, \cite{lotfi2024optimal}), there is still room for improvement through exact reductions, alternative solver encodings, and carefully qualified approximate methods. For D0L-systems, it is known that the problem can be solved in polynomial time if the size of the alphabet is fixed \cite{McQuillanBernardPrusinkiewicz2018}, but it is NP-complete if the alphabet is not fixed \cite{DuffyHillisKhanMcQuillanShan2025}. Prior work gives exact algorithms for fixed alphabet/context-size regimes and establishes NP-completeness for broader L-system inference settings. Our result is complementary: it gives an explicit characteristic-graph reduction to MIS/SAT for the finite exact-string D0LII formulation considered here.~\cite{McQuillanBernardPrusinkiewicz2018,DuffyHillisKhanMcQuillanShan2025} Heuristic systems such as PMIT-D0L have also been evaluated on standard L-system benchmarks; our purpose here is not to replace such systems, but to give an exact graph-theoretic reduction and decision route for the formal finite-string problem.~\cite{BernardMcQuillan2021} Our primary contribution is not a new guaranteed quantum algorithm for this NP-complete problem; it is the characteristic-graph reduction that makes the exact MIS/SAT formulation explicit. The quantum component is included only as an exploratory route for applying QAOA to the resulting MIS instances and is evaluated as a heuristic candidate generator.

This paper gives an exact graph-theoretic formulation of the D0L inductive inference problem (inferring a D0L-system from a sequence of strings, hereafter D0LII). We present a theorem that relates D0LII to the maximum independent set (MIS) problem by giving a polynomial-time and polynomial-size encoding of a D0LII instance into a graph, which we call the characteristic graph. We then show that the D0LII solution is equivalent to finding an independent set of a specific size in the characteristic graph. This encoding and the related theorem yield several solver routes. The first is a classical solution that combines an exact MIS computation on the characteristic graph with reconstruction of the production rules; by the equivalence theorem this is a genuine \emph{decision procedure} for D0LII: it returns a generating D0L-system when one exists, and reports nonexistence otherwise. The second is a quantum heuristic that solves the MIS approximately through the Quantum Approximate Optimization Algorithm (QAOA, \cite{farhi2014quantum,ruan2023quantum,choi2019tutorial}) and then attempts to decode production rules; because QAOA is approximate, it can fail to return a valid system and, in particular, cannot certify that no system exists. We therefore position the classical algorithm as the method that decides the problem and the quantum algorithm as an exploratory proof-of-concept. A third solution, which we will not discuss in detail but is evident, involves using a polynomial encoding of MIS into SAT and then employing modern SAT solvers.

This paper is organized as follows: Section \ref{def-l} provides the necessary background on L-systems, MIS, and quantum notation. Section \ref{def-encode-l-MIS-SAT} presents the characteristic-graph reduction to MIS and SAT, together with the examples in Section~\ref{worked-example} and the optional fixed-identity-symbol pruning. Section \ref{def-Quant-MIS} formulates the MIS QUBO and the QAOA heuristic. Section~\ref{sec-algorithms} separates the exact classical decision route from the QAOA-guided heuristic and gives the resource analysis that separates two-qubit gate count from circuit depth. Section~\ref{sec:experiments} reports computational validation on $22$ positive benchmark traces, $8$ negative controls, fixed-identity pruning, gate/depth accounting, and small statevector QAOA simulations. Section~\ref{sec:conclusion} concludes and discusses future directions.

\section{Preliminaries and Notations}\label{def-l}

This section introduces the foundational concepts and notations required for our work. First, we define 0L-systems and D0L-systems, formally specify the MIS problem, and introduce quantum notations relevant to our approach. These preliminaries are needed to encode the D0LII into MIS.

In the study of formal language theory, an alphabet $V$ is a finite set of symbols. Furthermore, $V^*$ represents the set of all strings over $V$, which includes the empty word $\varepsilon$. A language over $V$ is any subset of $V^*$. Given any string $S$, we use the notation $|S|$ to refer to its length and $S[i]$ to refer to the $i$-th ($1\leq i \leq |S|$) character of $S$. Furthermore, for $i$ and $j$ where $1 \leq i \leq j \leq |S| + 1$, we define $S[i:j]$ as follows:
\begin{itemize}
    \item If $i = j$, then it is the empty string.
    \item If $i < j < |S| + 1$, then it is the substring from index $i$ (inclusive) to $j$ (exclusive).
    \item If $i < j = |S| + 1$, then it is the suffix starting from index $i$ (inclusive).
\end{itemize}

Please see \cite{hopcroft2001introduction} for an introduction to formal language theory and computational complexity. We also refer to the cardinality of a finite set $X$ as simply $|X|$, which represents the number of elements in $X$.

\subsection{0L-systems and D0L-systems}

Next, we introduce L-systems and related notations, following the approaches of \cite{McQuillanBernardPrusinkiewicz2018} and \cite{prusinkiewicz2012algorithmic}. A context-free L-system, also known as a 0L-system, applies productions to symbols independently of their context within a string. We denote an L-System as $G = (V, \omega, P)$, where $V$ is an alphabet, $\omega \in V^*$ is the axiom, and $P \subseteq V \times V^*$ is a finite set of productions.
In the literature, L-systems that do \emph{not} allow productions to the empty word (i.e., every production $a \to x$ satisfies $x \neq \varepsilon$) are called \textit{propagating}, or non-erasing. We write a production $(a, x) \in P$ as $a \rightarrow x$, where $a$ is called the predecessor, and $x$ is called the successor.

Typically, we assume that each $a \in V$ has at least one production $a \rightarrow x$ in $P$. However, for mathematical convenience, we may consider ``partial'' L-systems where this property does not hold. A 0L-system $G$ is \textit{deterministic} (and called a D0L-system) if each $a \in V$ has exactly one production with it as a predecessor.
For a string $\mu = a_1 \cdots a_n \in V^*$, we write $\mu \Rightarrow \nu$ to indicate that $\mu$ directly derives $\nu$ if $\nu = x_1 \cdots x_n$, where $a_i \rightarrow x_i \in P$ for all $1 \leq i \leq n$. A derivation $d$ in $G$ comprises:

\begin{enumerate}
\item A trace $(w_0, \ldots, w_m)$ where $w_{i} \Rightarrow w_{i+1}$ for $0 \leq i < m$.
\item A function $\sigma_{d}: \{(i,j) \mid 0 \leq i < m , 1\leq j \leq |w_i| \} \rightarrow P$ such that if $w_i = A_1\cdots A_{|w_i|}$ and $1\leq j \leq |w_i|$, then $w_{i+1} = \alpha_1\ldots \alpha_{|w_i|}$ where $\sigma_{d}(i,j) = A_{j} \rightarrow \alpha_j$.
\end{enumerate}
The function $\sigma_d$ specifies which productions are applied to each letter in the derivation.

A sequence $\theta=(w_0,\ldots,w_m)$ is said to be compatible with a 0L system $G$ if $G$ can generate $\theta$ as a trace; otherwise, it's incompatible.
With these definitions, we can now formally state D0LII:

\vspace{0.5em}
\begin{questionbox}
	D0L inductive inference problem or D0LII: Given alphabet $V$, and a sequence $\theta = (w_0, \ldots, w_m)$ of non-empty strings over $V$, construct a D0L $G = (V, \omega, P)$ such that $\theta$ is compatible with $G$, or decide that it is not possible.
\end{questionbox}
\vspace{0.5em}

\begin{remark}
We restrict attention to sequences of non-empty strings, i.e., $|w_i| \geq 1$ for all $0 \leq i \leq m$. This is a deliberate restriction on the scope of D0LII as defined here, not a claim that sequences with empty strings are uninteresting. Sequences containing empty words may still be compatible with D0L-systems: for example, $(a, \varepsilon, \varepsilon)$ is compatible via $a \rightarrow \varepsilon$. However, when $w_i = \varepsilon$ for some $0 \leq i < m$ and $w_{i+1} \neq \varepsilon$, the sequence is immediately incompatible with every D0L-system (since applying any morphism to $\varepsilon$ yields $\varepsilon$), and can be rejected without invoking the characteristic graph construction. Our focus on non-empty strings captures the instances studied in the existing D0L inference literature, such as those in~\cite{BernardMcQuillan2021}.
\end{remark}

\subsection{Maximum Independent Set (MIS) Problem}

In this part, we will define the MIS problem for graphs. In an undirected graph $G = (V(G), E(G))$, a set $I \subseteq V(G)$ is called an \textit{independent set}, if the induced graph on $I$ does not contain any edges; that is, there is no edge between any $u,v \in V(I)$. An independent set $I$ is termed a \emph{maximum independent set} if it satisfies the following condition:
\[
|I| = \max \left\{ |J| : J \subseteq V(G), J \text{ is independent} \right\}.
\]
The decision version of maximum independent set is NP-complete, and the corresponding optimization/search problem is NP-hard (\cite{garey1979computers}).

\subsection{Quantum Notations}
In the sections where we discuss the quantum algorithm for MIS, we will use the following notations, which are adopted from \cite{nielsen2001quantum}:

\begin{itemize}
\item \( Z_i \): Pauli-Z operator on qubit \( i \)
\item \( X_i \): Pauli-X operator on qubit \( i \)
\item \( R_{Z_i}(\theta) \): Rotation around Z-axis by angle \( \theta \) on qubit \( i \)
\item \( R_{X_i}(\theta) \): Rotation around X-axis by angle \( \theta \) on qubit \( i \)
\item \( \ket{+} = \frac{1}{\sqrt{2}}(\ket{0} + \ket{1}) \): The superposition state
\item \( \delta_{ij} \): Kronecker delta, defined as \( \delta_{ij} = \begin{cases} 1 & \text{if } i = j \\ 0 & \text{otherwise} \end{cases} \)
\item $I$ is the identity matrix, $\bar{1}$ is the all-ones vector, and $J = \bar{1}\bar{1}^\top$ is the all-ones matrix.
\end{itemize}

For implementing the cost Hamiltonian in QAOA, we also require two-qubit entangling operations:
\begin{itemize}
\item \( R_{Z_i Z_j}(\theta) = e^{-i\frac{\theta}{2} Z_i \otimes Z_j} \): A two-qubit rotation that applies a phase based on the parity of qubits \( i \) and \( j \). This gate can be decomposed into two CNOT gates and a single-qubit \( R_Z \) rotation as follows:
\[
R_{Z_i Z_j}(\theta) = \text{CNOT}_{i,j} \cdot (I \otimes R_{Z_j}(\theta)) \cdot \text{CNOT}_{i,j}.
\]
\end{itemize}

Having proposed the main question we intend to solve, we are now ready to encode the D0LII into MIS.

\section{Encoding D0LII into MIS and SAT}\label{def-encode-l-MIS-SAT}

In this section, we will discuss how to map an instance of D0LII into MIS and SAT in polynomial time.
One way to find an encoding from D0LII to SAT is to establish an encoding from D0LII to MIS. Since a polynomial reduction exists from MIS to SAT, denoted as $f$, it is sufficient to construct a polynomial reduction, such as $g$, from the D0LII to MIS. The reduction from MIS to SAT is known (\cite{garey1979computers}, and is implied from their NP-completeness). This relationship can be visualized as in Figure \ref{fig:dlip-mis-sat}.

\begin{figure}[htbp]
	\centering
	\begin{tikzpicture}[->,>=stealth',node distance=4cm, thick]
		\node (A) {D0LII};
		\node (B) [below left=of A] {MIS};
		\node (C) [below right=of A] {SAT};
		\path[every node/.style={font=\sffamily\small}]
		(A) edge node [midway, left] {$g$} (B)
		(A) edge node [midway, right] {$f \circ g$} (C)
		(B) edge node [above] {$f$} (C);
	\end{tikzpicture}
	\caption{Illustration of the relationship between D0LII, MIS, and SAT.}
	\label{fig:dlip-mis-sat}
\end{figure}
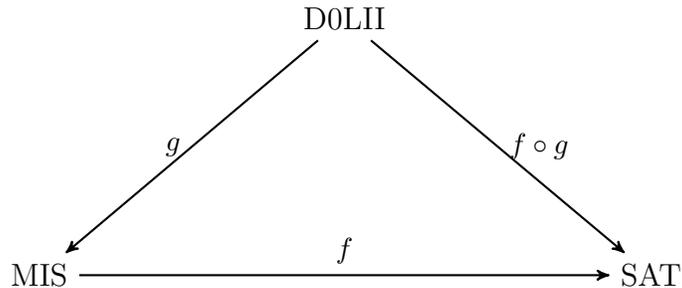

In order to find the encoding from D0LII to MIS, we will introduce the characteristic graph of a D0LII problem which will be a simple undirected graph. Such a graph encodes the information of $\theta$ and the restriction of D0LII into a single graph. Given the sequence $\theta=(w_0,\ldots,w_m)$, let $G_{\theta}$ be the characteristic graph of a D0LII with the following vertex set:

\begin{itemize}
	\item For $0 \leq i \leq m-1$ and $1 \leq j \leq |w_i|$, we construct complete graphs $G_{i,j}$ such that:
    \begin{itemize}
        \item If $j=1=|w_i|$:
            \[
        	V(G_{i,j}) = \{(i,j,1,|w_{i+1}|+1)\};
	    \]
        \item If $j=1$ and $j < |w_i|$:
            \[
        	V(G_{i,j}) = \{(i,j,1,\text{end}) : 1\leq \text{end} \leq |w_{i+1}|+1\};
	    \]
        \item If $1 < j$ and $j = |w_i|$:
            \[
        	V(G_{i,j}) = \{(i,j,\text{start},|w_{i+1}|+1) : 1 \leq \text{start} \leq |w_{i+1}|+1\};
	    \]
        \item If $1 < j$ and $j < |w_i|$:
            \[
        	V(G_{i,j}) = \{(i,j,\text{start},\text{end}) : 1 \leq \text{start} \leq \text{end} \leq |w_{i+1}|+1\}.
	    \]
    \end{itemize}
	\item We then define: \[V(G_{\theta}) = \bigcup\limits_{0\leq i<m,1\leq j\leq |w_{i}|} V(G_{ij}).\]
\end{itemize}

Hereafter, we define $E(G_{\theta})$, which includes all edges in
\[
\bigcup\limits_{0 \leq i < m,\ 1 \leq j \leq |w_i|} E(G_{ij}).
\]
Also let $(i_1, j_1) \neq (i_2, j_2)$ and $v, w \in V(G_{\theta})$ such that
\[
v \in V(G_{i_1, j_1}) \quad \text{and} \quad w \in V(G_{i_2, j_2}).
\]
Then we write
\[
v = (i_1, j_1, s_1, e_1) \quad \text{and} \quad w = (i_2, j_2, s_2, e_2).
\]
The edge $(v, w) \in E(G_{\theta})$ if at least one of the following conditions holds:
\begin{enumerate}
\item $w_{i_1}[j_1]=w_{i_2}[j_2]$ and $w_{i_1+1}[s_1:e_1] \neq w_{i_2+1}[s_2:e_2]$;
\item the two vertices correspond to adjacent positions in the same predecessor string and their intervals do not meet at the shared boundary. Equivalently, if $i_1=i_2$ and $j_2=j_1+1$, add an edge when $e_1\neq s_2$; if $i_1=i_2$ and $j_1=j_2+1$, add an edge when $e_2\neq s_1$.
\end{enumerate}
The construction can also be viewed as a D0L-specific incompatibility version of the microstructure graph used in finite constraint-satisfaction problems: predecessor positions are variables, candidate substrings are domain values, and graph edges encode forbidden pairs.~\cite{Jegou1993}

Having defined $G_{\theta}$, we are ready to present our theorem relating D0LII to MIS.

\begin{thm}\label{MIS-encoding}
    Assume $V = \{a_1, \ldots, a_n\}$ is an alphabet, and $\theta = (w_0, \ldots, w_m)$ is a sequence of non-empty strings over $V$ (i.e., $|w_i| \geq 1$ for all $0 \leq i \leq m$). Then there exists a D0L-system $G$ such that $\theta$ is compatible with $G$ if and only if $G_{\theta}$ has a MIS of size $k$, where:
	\[
	k = \sum_{i=0}^{m-1}\lvert w_i \rvert.
	\]
\end{thm}

\begin{proof}
	$(\Rightarrow)$: Suppose D0L-system $G$ with production set $P$ derives $\theta$. We will show that  $G_{\theta}$ has a MIS of size $k$. Since $G_{\theta}$ induced on $G_{i,j}$ is a complete graph, then any independent set has size at most $k$. Therefore, it remains to be shown that some independent set attains this upper bound.  $G$ derives $\theta$, so there exists a derivation $d$ in $G$ consisting of:

\begin{enumerate}
\item A trace $(w_0, \ldots, w_m)$ where $w_{i} \Rightarrow w_{i+1}$ for $0 \leq i < m$.
\item A function $\sigma_{d}: \{(i,j) \mid 0 \leq i < m , 1\leq j \leq |w_i| \} \rightarrow P$ such that if $w_i = A_1\cdots A_{|w_i|}$ and $1\leq j \leq |w_i|$, then $w_{i+1} = \alpha_1\ldots \alpha_{|w_i|}$ where $\sigma_{d}(i,j) = A_{j} \rightarrow \alpha_j$.
\end{enumerate}

Let \( Z_{i,j} = |\alpha_{j}| \), where \( \sigma_d(i,j) = A_j \rightarrow \alpha_j \) for \( 0 \leq i < m \), \( 1 \leq j \leq |w_i| \), and define:
\[
I = \{(i,j,s,e) \mid s = Z_{i,1} + \cdots + Z_{i,j-1} + 1, \, e = s + Z_{i,j}, \, 0 \leq i < m, \, 1 \leq j \leq |w_i| \}.
\]
Certainly, $|I| = k$. It remains to show that $I$ is an independent set.
Assume the contrary, there exists $(i_1,j_1,s_1,e_1),(i_2,j_2,s_2,e_2) \in I$ such that:
	\[\left((i_1,j_1,s_1,e_1),(i_2,j_2,s_2,e_2)\right) \in E(G_{\theta}).\]
	For such an edge to exist, the pair must satisfy at least one of the conditions for edges of $G_{\theta}$. This edge can not exist due to the first condition since if:
	\[w_{i_1}[j_1] = w_{i_2}[j_2],\]
	then having that $P$ is the production set of a D0L-system and the way $I$ is constructed, we get:
	\[w_{i_1+1}[s_1:e_1] = w_{i_2+1}[s_2:e_2],\]
	so condition~(1) does not hold.
	Therefore, we assume such an edge exists due to the second condition. Then we have one of the following two cases:
    \begin{enumerate}
        \item $i_1=i_2$ and $j_2 = j_1 + 1$ and $e_1 \neq s_2$;
        \item $i_1=i_2$ and $j_1 = j_2 + 1$ and $e_2 \neq s_1$.
    \end{enumerate}
     The first case has contradiction with $G$ deriving $\theta$, since from $w_i$ to $w_{i+1}$, the image of $w_{i_1}[j_1]$ is not coming right before the image of $w_{i_2}[j_2]$. Similarly the second case leads to contradiction as well. Therefore, neither of the conditions could be satisfied for
	\[\left((i_1,j_1,s_1,e_1),(i_2,j_2,s_2,e_2)\right)\]
	to be in $E(G_{\theta})$ which is a contradiction.
	Therefore, $I$ is a max independent set of size $k$.

	$(\Leftarrow)$: For the backwards direction we assume that $G_{\theta}$ has MIS $I$ of size $k$. We define the axiom of our D0L-system $G$ to be $\omega = w_0$. Since the induced subgraph to $G_{i,j}$ is complete, then from each $G_{i,j}$ there exists a unique vertex in $I$. We identify such vertex with $I_{i,j}$. It is left to build the production set $P$. Define:
	\[P' = \{w_i[j]\rightarrow w_{i+1}[s_{i,j}:e_{i,j}] \mid 0\leq i<m, 1\leq j\leq \lvert w_i\rvert, (i,j,s_{i,j},e_{i,j}) =  I_{i,j}\}.\]
	The set $P'$ defines productions for every symbol that appears in $w_0, \ldots, w_{m-1}$. For any symbol $a \in V$ that does not appear in any $w_i$ for $0 \leq i < m$, we assign the production $a \rightarrow a$. Let $P$ be the union of $P'$ and all such additional productions. This completes $P$ to cover all of $V$, and together with $\omega$, gives us a 0L-system over the given alphabet $V$.

	We first show that this 0L-system is a D0L-system. Each symbol that does not appear in $w_0,\ldots,w_{m-1}$ is assigned exactly one production by construction. For symbols appearing as predecessors in $P'$, assume the contrary: there exist $I_{i_1,j_1}, I_{i_2,j_2} \in I$ such that
\[
I_{i_1,j_1} = (i_1, j_1, s_{i_1,j_1}, e_{i_1,j_1}) \quad \text{and} \quad
I_{i_2,j_2} = (i_2, j_2, s_{i_2,j_2}, e_{i_2,j_2}),
\]
with the following conditions:
\[
w_{i_1}[j_1] \to w_{i_1+1}[s_{i_1,j_1}:e_{i_1,j_1}] \quad \text{and} \quad
w_{i_2}[j_2] \to w_{i_2+1}[s_{i_2,j_2}:e_{i_2,j_2}] \in P',
\]
such that
\[
w_{i_1+1}[s_{i_1,j_1}:e_{i_1,j_1}] \neq w_{i_2+1}[s_{i_2,j_2}:e_{i_2,j_2}],
\]
and
\[
w_{i_1}[j_1] = w_{i_2}[j_2].
\]
However, by the first condition of edges, we have:
\[(I_{i_1,j_1},I_{i_2,j_2}) \in E(G_{\theta}),\]
which is a contradiction with $I$ being an independent set. Therefore, the 0L-system $G$ is a D0L-system.

	It remains to show that $G = (V,\omega,P)$ derives $\theta$. More specifically, we need to show that for $0\leq i<m$, $G$ applied to $w_i$ will give us $w_{i+1}$. $G$ applied to $w_i$, will give us the following:

\[w_{i+1}[s_{i,1}:e_{i,1}]w_{i+1}[s_{i,2}:e_{i,2}]  \ldots w_{i+1}[s_{i,|w_i|}:e_{i,|w_i|}],\]
where the last string is the concatenation of the substrings $w_{i+1}[s_{i,j}:e_{i,j}]$ for $1 \leq j \leq |w_i|$.
From the construction of $V(G_{i,1})$ and $V(G_{i,|w_i|})$, we know that:
\[s_{i,1} = 1 \text{ and } e_{i,|w_i|} = |w_{i+1}|+1.\]
Furthermore, since $I_{i,j}$ are vertices of an independent set, and by the second condition of edges $G_{\theta}$, we get that for $1\leq j< |w_i|$ we have:
\[e_{i,j} = s_{i,j+1},\]
finishing the proof that $G$ from $w_i$ derives $w_{i+1}$. Since $i$, $0\leq i<m$ was arbitrary, $G$ derives $\theta$ finishing the proof of the backward direction.
\end{proof}

\begin{cor}\label{MIS-encoding-cor-1}
	Given \(\theta=(w_0,\ldots,w_m)\) with all strings non-empty, if  \(G_{\theta}\) has a MIS \(I\) of size \(k\) (where \(k\) is as defined in the statement of the theorem) in \(G_{\theta}\), then the D0L-system \((V, w_0, P)\) generates \(\theta\) where:
	\begin{enumerate}
		\item \(V = \{w_{i}[j] \mid 0 \leq i \leq m, 1 \leq j \leq |w_{i}|\}\), i.e., all symbols appearing anywhere in $\theta$;
		\item \(P = \{w_i[j] \rightarrow w_{i+1}[s_{i,j}:e_{i,j}] \mid 0 \leq i < m, 1 \leq j \leq |w_i|, (i,j,s_{i,j},e_{i,j}) \in I\}\),
		\item For any symbol $a \in V$ that does not appear as $w_i[j]$ for any $0 \leq i < m$, the production $a \rightarrow a$ is added to $P$.
	\end{enumerate}
\end{cor}

\begin{cor}\label{MIS-encoding-cor-2}
	The D0LII can be polynomially encoded into an MIS by constructing \(G_{\theta}\), and the production rules can be extracted from the vertices selected by the algorithm solving the MIS. By utilizing the polynomial encoding of MIS to SAT, the D0LII can be polynomially encoded into SAT.
\end{cor}

\begin{remrk}[Fixed identity symbols and induced-subgraph pruning]\label{rem:fixed-identity}
In many graphical L-systems, some symbols are drawing or control commands whose interpretation is fixed rather than inferred. Suppose a subset $C\subseteq V$ is known in advance to satisfy $c\mapsto c$ for every $c\in C$. Let $G_{\theta,C}$ be the induced subgraph of $G_\theta$ obtained by deleting every vertex $(i,j,s,e)$ with $w_i[j]\in C$ and $w_{i+1}[s:e]\neq w_i[j]$. Then the proof of Theorem~\ref{MIS-encoding} restricts immediately to this induced subgraph: there exists a compatible D0L-system with $c\mapsto c$ for all $c\in C$ if and only if $G_{\theta,C}$ has an independent set of size $k$. This optional pruning is not needed for the correctness of the main theorem, but when the input format declares turtle or drawing constants, it removes interval choices that are incompatible with the known identity productions and avoids counting spurious systems, such as erasing a declared drawing symbol.
\end{remrk}

\subsection{Examples}\label{worked-example}

\begin{exmp}[The trace $\theta=(ab,aab)$]\label{ex:ab-aab}
We illustrate the full pipeline with a concrete example. Let $\theta = (ab, aab)$, so $m = 1$, $w_0 = ab$, $w_1 = aab$, and $V = \{a, b\}$.

\textbf{Step 1: Parameters.} We compute $k = |w_0| = 2$.

\textbf{Step 2: Vertex sets.} We have $|w_0| = 2$ and $|w_1| = 3$.
\begin{itemize}
    \item $G_{0,1}$: since $j = 1 < |w_0| = 2$, we get
    \[V(G_{0,1}) = \{(0,1,1,\mathrm{end}) : 1 \leq \mathrm{end} \leq 4\}.\]
    The four vertices correspond to productions $a \rightarrow \varepsilon$, $a \rightarrow a$, $a \rightarrow aa$, and $a \rightarrow aab$ respectively.
    \item $G_{0,2}$: since $j = 2 = |w_0|$ and $j > 1$, we get
    \[V(G_{0,2}) = \{(0,2,\mathrm{start},4) : 1 \leq \mathrm{start} \leq 4\}.\]
    The four vertices correspond to productions $b \rightarrow aab$, $b \rightarrow ab$, $b \rightarrow b$, and $b \rightarrow \varepsilon$ respectively.
\end{itemize}

\textbf{Step 3: Edges.} Within each clique $G_{0,1}$ and $G_{0,2}$, all pairs are connected (6 intra-clique edges each). Between the two cliques, since $w_0[1] = a \neq b = w_0[2]$, Condition~1 produces no inter-clique edges. Condition~2 ($j_2 = j_1 + 1 = 2$, $e_1 \neq s_2$) adds an edge between $(0,1,1,e_1)$ and $(0,2,s_2,4)$ whenever $e_1 \neq s_2$.

\textbf{Step 4: MIS.} An independent set of size $k = 2$ must select exactly one vertex from each of $G_{0,1}$ and $G_{0,2}$. By Condition~2, the pair $\{(0,1,1,e_1), (0,2,s_2,4)\}$ is independent if and only if $e_1 = s_2$. This precisely enforces that the image of $a$ and the image of $b$ concatenate to form $w_1 = aab$. The four valid independent sets of size 2 are:
\[
\begin{gathered}
\{(0,1,1,1),(0,2,1,4)\},\quad \{(0,1,1,2),(0,2,2,4)\},\\
\{(0,1,1,3),(0,2,3,4)\},\quad \{(0,1,1,4),(0,2,4,4)\},
\end{gathered}
\]
yielding four D0L-systems: $\{a\!\to\!\varepsilon, b\!\to\!aab\}$, $\{a\!\to\!a, b\!\to\!ab\}$, $\{a\!\to\!aa, b\!\to\!b\}$, and $\{a\!\to\!aab, b\!\to\!\varepsilon\}$. One can verify directly that applying each to $w_0 = ab$ yields $w_1 = aab$. In this example, since every symbol in $V = \{a, b\}$ appears as a predecessor in $w_0$, all productions are fully constrained by the choice of MIS, and the number of compatible D0L-systems over $V$ equals the number of MIS of size $k$ in $G_\theta$. Note that this equality holds here specifically because there are no symbols with unconstrained productions; in general, when symbols appear in $V$ but not as predecessors in $w_0, \ldots, w_{m-1}$, they receive the identity production $a \rightarrow a$ as a canonical completion, and the count of compatible D0L-systems over $V$ may exceed the number of MIS of size $k$.
\end{exmp}

\begin{exmp}[The trace $\theta=(ab,abab)$]\label{ex:ab-abab}
As a second, slightly larger instance for the QAOA diagnostics in Table~\ref{tab:qaoa}, we use the trace $\theta=(ab,abab)$, with $m=1$, $w_0=ab$, $w_1=abab$, and $V=\{a,b\}$.
\end{exmp}

\section{QAOA Heuristic for MIS on \texorpdfstring{$G_\theta$}{G-theta}}\label{def-Quant-MIS}

There have been recent developments in quantum computing regarding the MIS problem, specifically, the introduction of a quantum polynomial approximate algorithm by \cite{yu2021quantum} and a quantum algorithm with time complexity $\mathcal{O}(1.1488^n)$ \cite{dorn2005quantum}. The latter offers Grover-based algorithms for both maximal independent set and maximum independent set problems. However, since the oracle in the algorithm was not clearly specified, we developed our own approach more similar to \cite{yu2021quantum}. The approach we will practice here for solving MIS is based on turning the MIS into a QUBO (Quadratic Unconstrained Binary Optimization) matrix and then utilizing a QAOA to find an approximate solution. To see this, first note that for a graph $G$ with $n$ vertices, any candidate for an independent set for this graph can be represented by a vector of size $n$ such as $x = (x_1,\ldots,x_n)$ where $x_i \in \{0,1\}$ for $1\leq i\leq n$. Furthermore, let \( Q \) be an \( n \) by \( n \) matrix such that:

\begin{itemize}
\item \( Q_{i,j} = -1 \) if \( i = j \);
\item \( Q_{i,j} = 1 \) if \( (i, j) \in E(G) \);
\item \( Q_{i,j} = 0 \) if \( (i, j) \notin E(G) \) and $i\ne j$.
\end{itemize}

Then, we see in the following proposition that an MIS induces a vector $x$ that minimizes $x^{\top} Q x$, where $x^{\top}$ is the transpose of the column vector $x$.

\begin{prp} \label{prop:mis-minimization}
Given a graph \( G = (V(G), E(G)) \) with \( V(G) = \{v_1, \ldots, v_n\} \) and a matrix \( Q \) as the corresponding QUBO matrix, an MIS of \( G \) will minimize the function \( f(x) = x^\top Q x \).
\end{prp}

\begin{proof}
First, notice that for binary $x_i \in \{0,1\}$ (so $x_i^2 = x_i$):
\[
f(x) = \sum_{i} Q_{i,i}x_i + \sum_{(i,j) \in E(G)} 2Q_{i,j}x_i x_j = -\sum_i x_i + 2\sum_{(i,j)\in E(G)} x_i x_j.
\]
Given \( T \subset V(G) \), let \( x_T \) denote the corresponding binary vector for \( T \).

\begin{clm} \label{clm:reduce-nonindependent}
For any set \( T \) that is not independent, there exists \( S \subset T \) such that \( f(x_{S}) < f(x_{T}) \).
\end{clm}

\begin{cpf}
Let \( T \) be a non-independent set, and take \( a, b \in V(T) \) such that \( (a, b) \in E(T) \). Define \( S \) as the induced graph on \( V(T) \setminus \{b\} \). Then:
\[
f(x_T) = -|V(T)| + 2|E(T)|,
\]
and
\begin{align}
  f(x_S) &\leq -(|V(T)| - 1) + 2(|E(T)| - 1) = -|V(T)| + 2|E(T)| - 1 < f(x_T). \nonumber\qedhere
\end{align}
\end{cpf}

\begin{clm} \label{clm:reduce-to-independent}
For any set \( T \) that is not independent, there exists \( S \subset T \) such that \( S \) is independent and \( f(x_{S}) < f(x_{T}) \).
\end{clm}

\begin{cpf}
Apply Claim~\ref{clm:reduce-nonindependent} repeatedly until an independent set $S$ is reached. At each step $f$ decreases, giving $f(x_S) < f(x_T)$.
\end{cpf}

\begin{clm} \label{clm:independent-set-comparison}
Given two independent sets \( M \) and \( N \) such that \( |V(N)| < |V(M)| \), we have:
\[
f(x_M) < f(x_N).
\]
\end{clm}

\begin{cpf}
Since $M$ and $N$ are independent, $|E(M)|=|E(N)|=0$, so $f(x_M) = -|V(M)| < -|V(N)| = f(x_N)$.
\end{cpf}

Finally, let $M$ be any MIS of $G$ and $T$ any subset. If $T$ is independent, apply Claim~\ref{clm:independent-set-comparison}. Otherwise, by Claim~\ref{clm:reduce-to-independent}, find independent $N \subset T$ with $f(x_N) < f(x_T)$. Since $M$ is an MIS, $|V(N)| \leq |V(M)|$, so $f(x_M) \leq f(x_N) < f(x_T)$.
\end{proof}

\subsection{QUBO Formulation and Hamiltonian for D0LII}
We use the standard unconstrained MIS QUBO/Ising encoding.~\cite{Lucas2014,GloverKochenbergerHennigDu2022} The D0L-specific observation is that the clique cover of $G_\theta$ already upper-bounds independent sets by $k$, so no additional global prescribed-size penalty is needed when testing whether this bound is attained.

\begin{remrk}
Unlike general MIS problems, the characteristic graph $G_\theta$ has a special structure: it is a
union of complete graphs $G_{i,j}$, one per symbol position $(i,j)$. Since each $G_{i,j}$ is a clique,
any independent set contains \emph{at most one} vertex from each $G_{i,j}$. There are exactly
$k = \sum_{i=0}^{m-1}|w_i|$ such cliques, so the maximum independent set of $G_\theta$ has size
\emph{at most} $k$. By Theorem~\ref{MIS-encoding}, this maximum is achieved (size exactly $k$)
if and only if a compatible D0L-system exists. Consequently, the unconstrained MIS objective
$f(x) = x^\top Q x$ from Proposition~\ref{prop:mis-minimization} already targets independent sets of size
exactly $k$ when a solution exists; no additional size-penalty term is needed. This is in contrast
to general MIS instances, where a penalty $\lambda(\sum_i x_i - k)^2$ would be required to
enforce a prescribed size. For $G_\theta$, adding such a penalty is both unnecessary and harmful:
penalty-based Hamiltonians may distort the
optimization landscape and impede convergence~\cite{blekos2024review}.
\end{remrk}

We therefore use the standard unconstrained MIS cost function:
\[
\mathrm{cost}(x) = x^\top Q x = -\sum_{i=1}^n x_i + 2\sum_{(i,j)\in E(G_\theta)} x_i x_j,
\]
and map it to a quantum Hamiltonian by replacing each binary variable $x_i$ with
$\frac{1}{2}(I - Z_i)$:
\[
H_C = \frac{1}{4}\sum_{i,j} Q_{ij}(I - Z_i)(I - Z_j).
\]
Expanding and collecting terms:
\begin{equation}\label{hc}
H_C = \sum_{(i,j)\in E(G_\theta)} \frac{1}{2} Z_i Z_j
+ \frac{1}{2}\sum_{i=1}^n \bigl(1 - \deg(i)\bigr) Z_i + \text{constant},
\end{equation}
where $\deg(i)$ is the degree of vertex $i$ in $G_\theta$. The constant is dropped. The mixing
Hamiltonian is:
\begin{equation}\label{hm}
H_M = \sum_{i=1}^n X_i.
\end{equation}

\begin{remrk}
$H_C$ contains $Z_iZ_j$ interactions only for pairs $(i,j)\in E(G_\theta)$. Thus one cost layer uses $\mathcal{O}(|E(G_\theta)|)$ two-qubit gates and avoids the additional all-to-all interactions that would be introduced by a dense global size-penalty term. This should not be read as a claim that $G_\theta$ is sparse as a graph: repeated-symbol instances can be clique-dense, as reflected in the resource analysis below.
\end{remrk}

One layer of the QAOA circuit implements:
\begin{itemize}
\item $e^{-i\gamma H_C} = \prod_{(i,j)\in E(G_\theta)} R_{Z_iZ_j}(\gamma) \cdot
      \prod_{i=1}^n R_{Z_i}\!\bigl(\gamma(1-\deg(i))\bigr)$,
\item $e^{-i\beta H_M} = \prod_{i=1}^n R_{X_i}(2\beta)$.
\end{itemize}

For $p$ layers the circuit is $\mathrm{circ} = e^{-i\beta_p H_M}e^{-i\gamma_p H_C}\cdots e^{-i\beta_1 H_M}e^{-i\gamma_1 H_C}$.

Algorithm~\ref{alg1} presents the complete QAOA-based MIS solver applied to $G_\theta$.

\begin{algorithm}[!htbp]
\caption{QAOA for MIS on $G_\theta$}
\label{alg1}
\begin{algorithmic}[1]
\Procedure{QAOA.MIS.Solver}{$G,p$}
\State Form matrix $Q$ from $G$
\State Form $H_C$ and $H_M$ from Equations \ref{hc} and \ref{hm}
\State Initialize $\vec{\gamma} = (\gamma_1, \ldots, \gamma_p)$ and $\vec{\beta} = (\beta_1, \ldots, \beta_p)$ randomly
\State Set learning rate $\eta$ and number of iterations $T$
\State $\text{best\_cost} \gets \infty$
\State $\text{best\_set} \gets \emptyset$
\For{$t = 1$ to $T$}
    \State $\text{cost}_\text{avg} \gets 0$
    \State $e \gets \text{number of experiments}$
    \For{$i = 1$ to $e$}
        \State Prepare initial state $\ket{+}^{\otimes n}$
        \For{$l = 1$ to $p$}
            \State Apply $e^{-i \gamma_l H_C}$
            \State Apply $e^{-i \beta_l H_M}$
        \EndFor
        \State Measure the state to obtain bit string $x$
        \State Calculate $\mathrm{cost}(x) = x^\top Q x = -\sum_i x_i + 2\sum_{(i,j)\in E(G)}x_ix_j$
        \State $\text{cost}_\text{avg} \gets \text{cost}_\text{avg} + \text{cost}(x)$
        \If{$\text{cost}(x) < \text{best\_cost}$}
            \State $\text{best\_cost} \gets \text{cost}(x)$
            \State $\text{best\_set} \gets \{i : x_i = 1\}$
        \EndIf
    \EndFor
    \State $\text{cost}_\text{avg} \gets \text{cost}_\text{avg} / e$
    \State Estimate gradient $\nabla_{\vec{\gamma},\vec{\beta}} \text{cost}_\text{avg}$ using quantum gradient methods (see~\cite{wierichs2022general})
    \State Update $\vec{\gamma} \gets \vec{\gamma} - \eta \nabla_{\vec{\gamma}} \text{cost}_\text{avg}$
    \State Update $\vec{\beta} \gets \vec{\beta} - \eta \nabla_{\vec{\beta}} \text{cost}_\text{avg}$
\EndFor
\State \textbf{return} $\text{best\_set}$
\EndProcedure
\end{algorithmic}
\end{algorithm}

\section{Quantum and Classical Algorithms for D0LII}\label{sec-algorithms}

Having Theorem \ref{MIS-encoding} and Corollary \ref{MIS-encoding-cor-1}, we are now ready to present our QAOA hybrid and classical Algorithms, which are Algorithms \ref{alg:opt-l} and \ref{alg:opt-2} respectively. Both algorithms take the formal input alphabet $V$ as an explicit parameter, then extend it by unioning with all symbols appearing anywhere in $\theta = (w_0, \ldots, w_m)$; symbols in $V$ that do not appear as predecessors in $w_0, \ldots, w_{m-1}$ receive the identity production $a \rightarrow a$ as a canonical completion. The axiom is $w_0$. We then construct $G_{\theta}$. Algorithm \ref{alg:opt-l} provides an approximate solution to the MIS (based on Algorithm \ref{alg1}); because QAOA is approximate, it calls ValidateRepair and may return $\bot$ if repair fails. Algorithm \ref{alg:opt-2} finds the MIS of $G_{\theta}$ exactly: by Theorem \ref{MIS-encoding}, if $|I| < k$ the answer to D0LII is negative, and otherwise the L-system is constructed as in the backward direction of the proof.

It is important to be precise about which of these two algorithms \emph{decides} D0LII. Algorithm~\ref{alg:opt-2} is a genuine decision procedure: it computes an exact maximum independent set of $G_\theta$, and by Theorem~\ref{MIS-encoding} the comparison $|I| \overset{?}{=} k$ is logically equivalent to the existence of a compatible D0L-system. Hence Algorithm~\ref{alg:opt-2} returns a generating system when one exists and correctly returns $\mathrm{False}$ (a certificate of nonexistence) when one does not. Algorithm~\ref{alg:opt-l}, by contrast, is a heuristic search and \emph{not} a decision procedure. QAOA is an approximate optimizer: when it fails to find an independent set of size $k$ it does not establish that none exists, since the failure may simply reflect insufficient circuit depth, too few samples, or an unfavorable optimization landscape. The symbol $\bot$ returned by \textsc{ValidateRepair} therefore means ``this heuristic run did not recover a valid system,'' which is strictly weaker than ``no valid system exists.'' Only the exact route (Algorithm~\ref{alg:opt-2}, or equivalently an exact SAT solver on the SAT encoding) can certify nonexistence. We make this asymmetry explicit because it is easy to misread Algorithm~\ref{alg:opt-l} as a decision procedure; it should instead be understood as an approximate solver whose successes are meaningful but whose failures are inconclusive. The experiments in Section~\ref{sec:experiments} bear this out: the exact algorithm decides all benchmark instances in both directions (existence and nonexistence), whereas the QAOA route is exercised only as a proof-of-concept on small instances.

\begin{remrk}[QAOA output and post-processing]\label{rem:qaoa-postprocessing}
Because Algorithm~\ref{alg:opt-l} uses an approximate QAOA routine, the returned set $I$ may fail to encode a valid D0LII solution and must be post-processed before use; see Algorithm~\ref{alg:repair} below.
\end{remrk}

\begin{algorithm}[!htbp]
	\caption{QAOA-Guided Heuristic Search for Compatible D0L Candidates}
	\label{alg:opt-l}
	\begin{algorithmic}[1]
		\Procedure{QuantInferD0L}{$V, w_0, w_1, \dots, w_m, p$}
		\State $V \gets V \cup \{w_{i}[j] \mid 0 \leq i \leq m , 1 \leq j \leq |w_{i}|\}$
		\textit{// ensures V covers all observed symbols; extras get $a\to a$ below}
		\State $\omega \gets w_0$
		\State $P \gets \{\}$
		\State $\theta \gets (w_0, w_1, \dots, w_m)$
		\State Construct $G_{\theta}$
		\State $k \gets \sum_{i=0}^{m-1}\lvert w_i \rvert$
        \Statex \textit{For higher accuracy, increase $p$}
		\State $I \gets \text{QAOA.MIS.Solver}(G_{\theta},p)$
		\State $I \gets \text{ValidateRepair}(I, G_{\theta}, k, m, w_0, \ldots, w_m, V)$ \quad \textit{// post-process; see Algorithm~\ref{alg:repair}}
		\If{$I = \bot$}
		\State \textbf{return} $\bot$ \Comment{Heuristic failure; inconclusive, not nonexistence}
		\EndIf
		\For{$(i,j,s,e) \in I$}
		\State Add $w_i[j]\rightarrow w_{i+1}[s:e]$ to $P$
		\EndFor
		\For{each $a \in V$ not appearing as $w_i[j]$ for any $0 \leq i < m$}
		\State Add $a \rightarrow a$ to $P$
		\EndFor
		\State \textbf{return} $(V, \omega, P)$
		\EndProcedure
	\end{algorithmic}
\end{algorithm}

\begin{algorithm}[!htbp]
	\caption{Solving D0LII Using an MIS Classical Solver}
	\label{alg:opt-2}
	\begin{algorithmic}[1]
		\Procedure{ClassicalD0LSolver}{$V, w_0, w_1, \dots, w_m$}
		\State $V \gets V \cup \{w_{i}[j] \mid 0 \leq i \leq m , 1 \leq j \leq |w_{i}|\}$
		\State $\omega \gets w_0$
		\State $P \gets \{\}$
		\State $\theta \gets (w_0, w_1, \dots, w_m)$
		\State Construct $G_{\theta}$
		\State $I \gets \mathrm{MIS\_solver}(G_{\theta})$
		\State $k \gets \sum_{i=0}^{m-1}\lvert w_i \rvert$
		\If{$|I|<k$}
		\State \textbf{return} $\mathrm{False}$
		\EndIf
		\For{$(i,j,s,e) \in I$}
		\State Add $w_i[j]\rightarrow w_{i+1}[s:e]$ to $P$
		\EndFor
		\For{each $a \in V$ not appearing as $w_i[j]$ for any $0 \leq i < m$}
		\State Add $a \rightarrow a$ to $P$
		\EndFor
		\State \textbf{return} $(V, \omega, P)$
		\EndProcedure
	\end{algorithmic}
\end{algorithm}

\subsection{Post-Processing and Validation for the Quantum Algorithm}

A key issue with approximate quantum solvers is that the returned bit string $x$ may not satisfy the structural requirements needed to decode a valid D0L-system. Specifically, the returned set $I$ may:
\begin{enumerate}[label=(\alph*)]
    \item Contain two or more vertices from the same clique $G_{i,j}$ (violating independence within $G_{i,j}$);
    \item Contain no vertex from some clique $G_{i,j}$ (leaving a production undefined);
    \item Assign the same symbol two different images (violating the D0L determinism requirement).
\end{enumerate}
We address these issues with the following repair procedure.

\begin{algorithm}[!htbp]
\caption{ValidateRepair for QAOA Output}
\label{alg:repair}
\begin{algorithmic}[1]
\Procedure{ValidateRepair}{$I, G_{\theta}, k, m, w_0, \ldots, w_m, V$}
\State $I_{\mathrm{new}} \gets \emptyset$
\For{each $(i,j)$ with $0 \leq i < m$, $1 \leq j \leq |w_i|$ in lexicographic order}
    \State $S \gets I \cap V(G_{i,j})$ \quad \textit{// sampled vertices in clique $G_{i,j}$}
    \State $C \gets \{v \in S \mid v \text{ is non-adjacent to every vertex in } I_{\mathrm{new}}\}$
    \If{$C \neq \emptyset$}
        \State Add the lexicographically first vertex of $C$ to $I_{\mathrm{new}}$
    \Else
        \State $D \gets \{v \in V(G_{i,j}) \mid v \text{ is non-adjacent to every vertex in } I_{\mathrm{new}}\}$
        \If{$D = \emptyset$}
            \State \textbf{return} $\bot$
        \Else
            \State Add $\arg\min_{v \in D} \deg_{G_\theta}(v)$ to $I_{\mathrm{new}}$; break ties lexicographically
        \EndIf
    \EndIf
\EndFor
\If{$|I_{\mathrm{new}}| = k$ and $I_{\mathrm{new}}$ is independent in $G_\theta$}
    \State Decode the productions induced by $I_{\mathrm{new}}$
    \If{the decoded productions generate $w_{i+1}$ from $w_i$ for all $0\leq i<m$}
        \State \textbf{return} $I_{\mathrm{new}}$
    \EndIf
\EndIf
\State \textbf{return} $\bot$ \Comment{Approximation failed; no valid D0L-system recovered}
\EndProcedure
\end{algorithmic}
\end{algorithm}

Note that Algorithm~\ref{alg:repair} is a deterministic greedy repair heuristic. It scans the cliques in a fixed order, keeps sampled vertices when they remain compatible with the choices already made, and otherwise inserts a low-degree feasible replacement when one exists. If this process succeeds through all cliques, then $I_{\mathrm{new}}$ has one vertex from each $G_{i,j}$ and hence $|I_{\mathrm{new}}|=k$. The final independence and derivation checks are essential: they verify that the extracted productions form a consistent D0L-system and that the system maps every $w_i$ to $w_{i+1}$. If any clique has no feasible replacement, or if the final validation fails, the procedure returns $\bot$. This return value indicates only failure of the heuristic repair, not nonexistence of a compatible D0L-system.
Similar to our algorithm, one could use the polynomial encoding of MIS to SAT and solve the D0LII problem using either quantum SAT solvers or classical SAT solvers. We will omit the details of this algorithm here since the idea is very similar to Algorithm \ref{alg:opt-l}.

\subsection{QAOA Resource Analysis}\label{subsec:resource}

For Algorithm \ref{alg:opt-l}, we provide an analysis of the computational resources required by the QAOA-based approach. We frame this as a resource analysis rather than a traditional time complexity analysis because the scaling behavior of QAOA's hyperparameters (specifically the circuit depth $p$, the number of samples $e$, and the number of optimization iterations $T$) required to achieve a given solution quality is not yet fully understood theoretically \cite{blekos2024review}.

Recall that $n = |V(G_\theta)| = \mathcal{O}(kl^2)$ is the number of qubits, and
$|E(G_\theta)| = \mathcal{O}(vh^2l^4)$ where $v=|V|$, $l=\max_{0\leq i\leq m}|w_i|$, and
\[
h = \max_{a\in V} \left|\{(i,j):0\leq i<m,\ 1\leq j\leq |w_i|,\ w_i[j]=a\}\right|
\]
is the maximum number of predecessor occurrences of any symbol. Because we use the standard unconstrained MIS Hamiltonian (no
penalty term), $H_C$ is edge-local on the characteristic graph: it contains one two-qubit interaction per edge of $G_\theta$ and no additional all-to-all cardinality-penalty interactions.
The quantum circuit for each QAOA iteration involves the following operations:
\begin{itemize}
    \item Constructing an edge-list representation of $Q$ requires $\mathcal{O}(n + |E(G_\theta)|) = \mathcal{O}(vh^2l^4)$ classical operations. We note that $Q$ need not be stored as a dense $n \times n$ matrix; since $Q_{ij} = 0$ for non-adjacent, non-diagonal pairs, only the diagonal and edge entries are stored, giving a representation of size $\mathcal{O}(n + |E(G_\theta)|)$.
    \item Preparing the initial superposition state $\ket{+}^{\otimes n}$ requires $\mathcal{O}(n) = \mathcal{O}(kl^2)$ single-qubit gates.
    \item Implementing the cost Hamiltonian $H_C$ requires applying the two-qubit gate $R_{Z_iZ_j}$ for each edge $(i,j)\in E(G_\theta)$, which amounts to $\mathcal{O}(|E(G_\theta)|) = \mathcal{O}(vh^2l^4)$ two-qubit operations per layer. This is a \emph{gate count}, not a circuit depth: it counts the $R_{Z_iZ_j}$ gates in one cost layer but does not account for which of them may be executed simultaneously (see the depth analysis below).
    \item Implementing the mixing Hamiltonian $H_M$ requires $\mathcal{O}(n) = \mathcal{O}(kl^2)$ single-qubit gates per layer.
\end{itemize}

It is important to distinguish the two-qubit gate count from the circuit depth. The quantity $\mathcal{O}(|E(G_\theta)|)$ above is the \emph{number} of two-qubit gates in one cost layer; the two-qubit \emph{depth} of that layer is the number of sequential rounds needed to execute those gates, which is generally far smaller because gates acting on disjoint qubit pairs run in parallel. Under all-to-all qubit connectivity, the gates of a single cost layer can be scheduled in $\chi'(G_\theta)$ parallel rounds, where $\chi'(G_\theta)$ is the edge-chromatic number of $G_\theta$ (a proper edge colouring assigns commuting, qubit-disjoint $R_{Z_iZ_j}$ gates to each colour). By Vizing's theorem $\Delta(G_\theta) \le \chi'(G_\theta) \le \Delta(G_\theta)+1$, where $\Delta(G_\theta)$ is the maximum degree. Hence the two-qubit depth of one cost layer is $\Theta(\Delta(G_\theta))$, and for $p$ layers the two-qubit depth is
\begin{equation}\label{eq:depth-correct}
\mathcal{O}\!\bigl(p\,\chi'(G_\theta)\bigr) \;=\; \mathcal{O}\!\bigl(p\,\Delta(G_\theta)\bigr)\,,
\end{equation}
which is in general much smaller than the gate-count expression $\mathcal{O}(p\,|E(G_\theta)|)$. The gap disappears only in degenerate schedules where essentially no two interactions can be parallelized; for graphs with many disjoint edges the scheduled depth can be much smaller than $|E(G_\theta)|$, while for the dense, clique-rich graphs $G_\theta$ the degree remains large. The quantity $\mathcal{O}(p\cdot vh^2l^4)$ bounds the two-qubit gate count, not the circuit depth.

This smaller depth does not, however, make $G_\theta$ hardware-friendly: its clique-dense structure remains demanding. Each clique $G_{i,j}$ is complete on its vertices, so every vertex of a clique of size $c$ has degree at least $c-1$; therefore
\begin{equation}\label{eq:delta-clique}
\Delta(G_\theta) \;\ge\; \omega(G_\theta) - 1\,,
\end{equation}
where $\omega(G_\theta)$ is the size of the largest clique. Interior nonconstant symbol positions, when present, have clique size on the order of $l^2$ (the squared maximum string length), so $\Delta(G_\theta)$, and with it the two-qubit depth in \eqref{eq:depth-correct}, can grow polynomially with the string length even though it is far below $|E(G_\theta)|$. Moreover, the depth bound \eqref{eq:depth-correct} assumes all-to-all connectivity. On hardware with bounded connectivity (for example heavy-hex lattices of fixed degree), realizing the interactions of a size-$c$ clique requires routing via SWAP networks, which adds depth that grows with $c$; the clique structure of $G_\theta$ is thus genuinely demanding for near-term devices. Table~\ref{tab:gate-depth} in Section~\ref{sec:experiments} quantifies this gap on concrete instances, reporting $|E(G_\theta)|$, $\omega(G_\theta)$, $\Delta(G_\theta)$, and the number of rounds produced by the deterministic greedy edge-colouring used in the experiments. This is an explicit reproducible all-to-all schedule upper bound, not a claim to compute the optimal edge-chromatic number $\chi'(G_\theta)$; the gate-count-to-greedy-depth ratio is already one to two orders of magnitude for moderate instances. Native MIS platforms such as Rydberg-atom arrays on unit-disk graphs sidestep the SWAP overhead and are a more promising long-term target, which we discuss in the conclusion.

For $p$ layers, the total two-qubit gate count is therefore $\mathcal{O}(p \cdot vh^2 l^4)$, while the corresponding all-to-all two-qubit circuit depth admits the existential bound $\mathcal{O}(p\,\Delta(G_\theta))$ from~\eqref{eq:depth-correct}.

The classical optimization and evaluation costs are as follows:
\begin{itemize}
    \item Measurement and cost evaluation: We repeat the measurement $e$ times. Evaluating $\mathrm{cost}(x) = x^\top Qx$ for each bit string requires $\mathcal{O}(n + |E(G_\theta)|) = \mathcal{O}(vh^2l^4)$ operations, giving a total of $\mathcal{O}(e \cdot vh^2l^4)$ per iteration.
    \item Parameter optimization: Gradient estimation for the $2p$ variational parameters uses quantum gradient methods~\cite{wierichs2022general}. The exact per-parameter cost depends on the spectral properties of the generators $H_C$ and $H_M$ and the chosen differentiation scheme; deriving a tight bound is beyond the scope of this resource analysis.
\end{itemize}

Combining these, the displayed $\mathcal{O}(T \cdot p \cdot e \cdot vh^2 l^4)$ term represents the circuit-and-sampling cost over $T$ iterations:
\[
\mathcal{O}\!\left(T \cdot p \cdot e \cdot vh^2 l^4\right).
\]
The full optimization cost includes additional gradient-estimation costs, which are not quantified here.

This expression should be interpreted only as an accounting of the work performed by one fixed heuristic run. It is not a complexity guarantee for D0LII and does not imply that Algorithm~\ref{alg:opt-l} solves the NP-complete search/decision problem with guaranteed success probability. The essential open question is how the depth $p$, the number of samples $e$, and the number of optimizer iterations $T$ must grow with the instance size to obtain a desired success probability. Issues such as barren plateaus in the optimization landscape and sampling complexity remain active areas of research \cite{farhi2014quantum,blekos2024review}. We therefore do not claim that Algorithm~\ref{alg:opt-l} achieves a proven polynomial-time quantum algorithm, approximation guarantee, or speedup over classical exact MIS/SAT/MILP methods.

What our analysis does establish is that the per-layer two-qubit gate count, the all-to-all scheduled depth \eqref{eq:depth-correct}, and the per-sample classical cost can be bounded explicitly for a chosen instance and chosen hyperparameters. This suggests the approach is in principle implementable on quantum hardware for very small instances, though a full assessment of near-term practicality would require hardware-specific compilation and noise studies, which we leave for future work; the clique-induced degree \eqref{eq:delta-clique} and the resulting routing overhead on bounded-connectivity hardware are the main obstacles. Moreover, the characteristic graph construction from Theorem \ref{MIS-encoding} provides a natural problem Hamiltonian that could be used with other quantum optimization paradigms, such as adiabatic quantum computing or quantum annealing, which may benefit from different theoretical guarantees or empirical performance characteristics.

The implementation of both quantum and classical D0LII solvers developed for this paper is available at: \href{https://github.com/alilotfi90/D0L-Quant-and-Classical-Solver}{GitHub}.

\section{Experimental Evaluation}\label{sec:experiments}

We provide empirical support for the foregoing claims. The purpose of the experiments is deliberately limited: we validate the characteristic-graph decoding pipeline, demonstrate that the exact route rejects incompatible inputs, quantify the distinction between two-qubit gate count and circuit depth, and test the QAOA formulation only by small noiseless statevector simulations. These experiments are not intended to show quantum advantage.

\paragraph{Benchmark suite and fixed constants.} We use $22$ D0L systems from the benchmark-style collection associated with Bernard and McQuillan~\cite{BernardMcQuillan2021}. For each system we generated the shortest prefix $\theta=(w_0,\ldots,w_m)$ in which every explicitly specified non-identity rule symbol appears at least once as a predecessor. When the benchmark file declared symbols as known turtle/drawing symbols, those symbols were fixed to identity productions and Remark~\ref{rem:fixed-identity} was applied. Observed symbols with no explicit production rule in the benchmark file were treated the same way. This convention is necessary for faithful benchmark interpretation: without these fixed constants, one counts spurious compatible systems such as rules that erase a drawing symbol. With the constants fixed, 013\_dragoncurve, for example, has a unique compatible system under this convention for $m=1,2,3,4,5$; an erasing interpretation of $F$ is not legitimate here, since $F$ is a declared constant.

\paragraph{Existence instances.} Table~\ref{tab:bench} reports the classical validation results. The column $|V(G_{\theta,C})|$ is the number of vertices after fixed-identity pruning. The implementation searches the same substring choices encoded by the characteristic graph, reconstructs a morphism, and then validates the reconstructed morphism on every transition of $\theta$. All $22$ traces are solved successfully and regenerate their input sequences. The column ``recovered'' compares the inferred morphism with the generating benchmark morphism. Entries below $100\%$ are not failures: they show that finite D0L inference can be underconstrained, so the formal goal is to recover a compatible D0L-system rather than necessarily the unique original generator.

\begin{table}[!htbp]
\centering
\footnotesize
\setlength{\tabcolsep}{3pt}
\caption{Classical validation on $22$ generated D0L traces. Known fixed-identity symbols are pruned as in Remark~\ref{rem:fixed-identity}.}
\label{tab:bench}
\resizebox{\textwidth}{!}{%
\begin{tabular}{lrrrrrrr}
\toprule
Instance & $m$ & lengths & $k$ & $|V(G_{\theta,C})|$ & status & valid & recovered \\
\midrule
010\_algae & 2 & 1-2-3 & 3 & 9 & success & True & 2/2 \\
011\_cantordust & 2 & 1-3-9 & 4 & 76 & success & True & 2/2 \\
012\_kochcurve & 1 & 1-9 & 1 & 1 & success & True & 1/1 \\
013\_dragoncurve & 2 & 2-6-14 & 8 & 258 & success & True & 2/2 \\
014\_fractalplant & 2 & 1-18-89 & 19 & 20,811 & success & True & 2/2 \\
015\_pythagorastree & 2 & 1-5-14 & 6 & 157 & success & True & 2/2 \\
016\_sierpenskitriangle & 2 & 1-7-25 & 8 & 1,072 & success & True & 2/2 \\
017\_sierpenskitriangle2 & 1 & 3-13 & 3 & 133 & success & True & 2/2 \\
018\_gospercurve1 & 2 & 1-15-113 & 16 & 39,670 & success & True & 2/2 \\
020\_ecurve & 2 & 1-43-1093 & 44 & 14,380,234 & success & True & 2/2 \\
021\_peano & 2 & 1-97-2497 & 98 & 71,838,698 & success & True & 2/2 \\
022\_appletwig & 4 & \shortstack[r]{21-105-1123\\-4380-11505} & 5629 & 31,229,645,212 & success & True & 7/9 \\
024\_bushlike\_simplified & 3 & 1-35-221-800 & 257 & 10,843,756 & success & True & 4/4 \\
025\_aphanocladia & 3 & 1-2-16-138 & 19 & 58,751 & success & True & 4/4 \\
026\_dipterosiphoniaV1 & 9 & \shortstack[r]{4-5-20-42-64-102\\-164-248-361-482} & 1010 & 8,211,001 & success & True & 22/28 \\
027\_dipterosiphoniaV2 & 3 & 1-5-27-173 & 33 & 138,092 & success & True & 5/5 \\
028\_ditira\_reptans & 3 & 1-2-15-60 & 18 & 13,412 & success & True & 4/4 \\
029\_ditira\_zonaricola & 3 & 1-2-17-117 & 20 & 49,468 & success & True & 4/4 \\
030\_herpopteros & 3 & 1-2-15-189 & 18 & 91,226 & success & True & 4/4 \\
031\_herposiphonia & 3 & 1-6-49-330 & 56 & 992,573 & success & True & 5/5 \\
032\_metamorphe & 3 & 1-2-26-290 & 29 & 510,709 & success & True & 5/5 \\
034\_tenuissimum & 13 & \shortstack[r]{4-5-20-47-76-115\\-170-245-336-442\\-556-676-792-918} & 3484 & 45,496,914 & success & True & 25/31 \\
\bottomrule
\end{tabular}
}
\end{table}

The largest pruned graph in Table~\ref{tab:bench} is the 022\_appletwig instance, with $31,229,645,212$ vertices after pruning and $300,648,083,462$ vertices before pruning. The exact implementation used here does not materialize these graphs; it solves the equivalent interval-consistency problem directly and then validates the reconstructed morphism. This reinforces the positioning of the paper: the exact classical/SAT route is the practical decision mechanism, while QAOA is only a small-instance candidate generator.

\paragraph{Nonexistence instances.} A decision procedure must also reject incompatible inputs. Table~\ref{tab:negative} gives five hand-built incompatible traces and three corrupted benchmark traces. Each corrupted trace is obtained by replacing the final symbol of the last string with a fresh symbol (\texttt{\#}, or \texttt{\$} if needed); for the three listed traces this alteration makes the sequence incompatible with every D0L-system. In all $8$ cases the exact search returns no compatible D0L-system. By Theorem~\ref{MIS-encoding}, this is equivalent to saying that $G_{\theta,C}$ has no independent set of size $k$; therefore the maximum independent set size is strictly below $k$. This is precisely the certificate that the QAOA heuristic cannot provide.

\begin{table}[!htbp]
\centering
\footnotesize
\setlength{\tabcolsep}{3pt}
\caption{Negative controls. Rows labelled \texttt{hand:} are hand-built incompatible traces. The exact route rejects all $8$ incompatible traces; a QAOA failure would not constitute such a certificate.}
\label{tab:negative}
\begin{tabular}{lrrl}
\toprule
Input & $k$ & $|V(G_{\theta,C})|$ & Exact decision \\
\midrule
hand: (ab, ab, ba) & 4 & 12 & no compatible D0L \\
hand: (abc, abc, cba) & 6 & 36 & no compatible D0L \\
hand: (aa, ab) & 2 & 6 & no compatible D0L \\
hand: (ab, abab, ab) & 6 & 28 & no compatible D0L \\
hand: (abc, abcabc, abc) & 9 & 90 & no compatible D0L \\
010\_algae corrupted & 6 & 42 & no compatible D0L \\
011\_cantordust corrupted & 4 & 76 & no compatible D0L \\
013\_dragoncurve corrupted & 8 & 256 & no compatible D0L \\
\bottomrule
\end{tabular}
\end{table}

\paragraph{Multiplicity of compatible systems.} The reduction also clarifies the non-uniqueness of finite inference. A compatible canonical morphism is a D0L morphism that generates $\theta$ and assigns the identity production $a\rightarrow a$ to every symbol that does not occur as a predecessor in $w_0,\ldots,w_{m-1}$. The multiplicity counts enumerate distinct compatible canonical morphisms, with the enumeration capped at $65$ for tractability. With fixed constants, many classical fractal benchmarks have a single compatible canonical morphism in the prefixes tested, including 010\_algae, 011\_cantordust, and 013\_dragoncurve. Other biological examples remain ambiguous even after every explicit rule symbol has appeared as a predecessor. Under this capped enumeration, 022\_appletwig has exactly $12$ compatible canonical morphisms for the prefix in Table~\ref{tab:bench}, 026\_dipterosiphoniaV1 has exactly $64$, and 034\_tenuissimum has at least $65$. This observation is a useful byproduct of the reduction: model selection among compatible systems can be treated separately, for example by preferring shorter productions or domain-specific biological constraints.

\paragraph{Gate count versus circuit depth.} Table~\ref{tab:gate-depth} accompanies the resource analysis of Section~\ref{subsec:resource}. For explicit small graphs, we report the two-qubit gate count per cost layer $|E(G_{\theta,C})|$, the maximum clique size $\omega$, the maximum degree $\Delta$, and the number of parallel rounds produced by a deterministic greedy proper edge-colouring. This reported number is an explicit achievable all-to-all two-qubit schedule for one cost layer, not necessarily the optimal edge-chromatic number $\chi'(G_{\theta,C})$, while $|E|$ is only a gate count. The gap is substantial: for the 016\_sierpenskitriangle row, one layer has $318,141$ logical two-qubit interactions but a greedy all-to-all schedule of $1,040$ rounds. The table also confirms that the interaction graph is hardware-unfriendly: because $\Delta\geq \omega-1$ and the graphs contain large cliques, sparse-connectivity gate-model hardware would require additional routing overhead.

\begin{table}[!htbp]
\centering
\footnotesize
\setlength{\tabcolsep}{3pt}
\caption{Two-qubit gate count $|E(G_{\theta,C})|$ versus a reproducible achievable all-to-all two-qubit schedule from deterministic greedy edge-colouring for one QAOA cost layer. The reported depth is an upper bound from the implemented schedule, not necessarily $\chi'(G_{\theta,C})$.}
\label{tab:gate-depth}
\begin{tabular}{lrrrrr}
\toprule
Instance & $n$ & $|E|$ & $\omega$ & $\Delta$ & depth \\
\midrule
Example~\ref{ex:ab-aab} & 8 & 24 & 4 & 6 & 6 \\
030\_herpopteros, $m=2$ & 33 & 495 & 16 & 31 & 42 \\
010\_algae, $m=3$ & 42 & 635 & 21 & 35 & 45 \\
011\_cantordust, $m=2$ & 76 & 2,679 & 55 & 73 & 94 \\
015\_pythagorastree, $m=2$ & 157 & 10,001 & 120 & 141 & 205 \\
013\_dragoncurve, $m=2$ & 258 & 16,594 & 120 & 228 & 228 \\
016\_sierpenskitriangle, $m=2$ & 1,072 & 318,141 & 351 & 709 & 1,040 \\
\bottomrule
\end{tabular}
\end{table}

\paragraph{QAOA statevector simulations.} Finally, Table~\ref{tab:qaoa} reports small noiseless statevector simulations of the QAOA candidate generator. The simulations use a seeded multi-start local parameter search with layerwise initialization: for depth $p>1$, one restart initializes the first $p-1$ layers from the best depth-$(p-1)$ parameters and appends a zero layer. No quantum hardware is used. The columns report the normalized approximation ratio, the probability that one measurement samples an optimal MIS bit string, and the fraction of sampled bit strings that decode to a valid D0L-system after the greedy \textsc{ValidateRepair} procedure. The approximation ratio is $r=(f_{\max}-\mathbb{E}[f(x)])/(f_{\max}-f_{\min})$, where $f_{\max}$ and $f_{\min}$ are the largest and smallest values of $f(x)=x^\top Qx$ over all bit strings. The multiplier in parentheses is $P(\mathrm{MIS})$ divided by the probability that a uniformly random bit string is an optimal MIS. The results should be interpreted conservatively. On these toy instances, the optimized circuits sample an optimal MIS more often than uniform random bit strings and the repair stage validates the decoded systems. The large increase in the pruned 013\_dragoncurve row is a feature of this very small eight-qubit diagnostic instance and of the chosen optimizer, not a scaling claim. Overall, these are small statevector diagnostics: the table does not establish asymptotic usefulness, a success guarantee, or an advantage over exact classical MIS/SAT/MILP solvers.

\begin{table}[!htbp]
\centering
\footnotesize
\setlength{\tabcolsep}{3pt}
\caption{Noiseless statevector QAOA simulations on small pruned characteristic graphs. These are implementation checks only; exact classical solvers handle these instances immediately.}
\label{tab:qaoa}
\begin{tabular}{lrrrcccc}
\toprule
Instance & $n$ & $|E|$ & $k$ & $p$ & appr. ratio & $P(\mathrm{MIS})$ & repair$\to$D0L \\
\midrule
Example~\ref{ex:ab-aab} & 8 & 24 & 2 & 1 & 0.960 & 0.110 (7$\times$) & 1.000 \\
 &  &  &  & 2 & 0.971 & 0.155 (10$\times$) & 1.000 \\
 &  &  &  & 3 & 0.972 & 0.179 (11$\times$) & 1.000 \\
\addlinespace
Example~\ref{ex:ab-abab} & 10 & 40 & 2 & 1 & 0.974 & 0.077 (16$\times$) & 1.000 \\
 &  &  &  & 2 & 0.981 & 0.090 (18$\times$) & 1.000 \\
 &  &  &  & 3 & 0.986 & 0.146 (30$\times$) & 1.000 \\
\addlinespace
010\_algae ($A,AB,ABA$) & 9 & 27 & 3 & 1 & 0.948 & 0.013 (7$\times$) & 1.000 \\
 &  &  &  & 2 & 0.959 & 0.018 (9$\times$) & 1.000 \\
 &  &  &  & 3 & 0.967 & 0.020 (10$\times$) & 1.000 \\
\addlinespace
013\_dragoncurve, $m=1$ pruned & 8 & 27 & 2 & 1 & 0.961 & 0.030 (8$\times$) & 1.000 \\
 &  &  &  & 2 & 0.986 & 0.454 (116$\times$) & 1.000 \\
 &  &  &  & 3 & 0.988 & 0.521 (133$\times$) & 1.000 \\
\bottomrule
\end{tabular}
\end{table}

\paragraph{Summary.} The exact characteristic-graph route validates all $22$ positive benchmark traces and rejects all $8$ negative controls. Fixed-identity pruning is essential for faithful treatment of turtle/control symbols and removes spurious ambiguity. The gate/depth table separates logical two-qubit interactions from schedulable circuit depth. The QAOA simulations serve as a transparent proof of concept, but the paper does not present QAOA as a decision procedure or as a practically competitive solver.

\section{Conclusion and Future Directions}\label{sec:conclusion}

In this paper, we introduced the characteristic graph for deterministic context-free L-system inference and proved an exact equivalence between compatible D0L-systems and independent sets of a prescribed size. This graph-theoretic reduction is the main contribution of the paper. It gives polynomial reductions from D0LII to MIS and SAT and yields an exact classical decision route when combined with an exact solver. Symbols not appearing as predecessors receive the identity production $a\rightarrow a$ as a canonical completion over the full observed alphabet.

We also formulated the associated MIS objective as a QUBO and described a QAOA-based candidate generator. We deliberately limit the quantum claim: the QAOA procedure is a heuristic, not a decision procedure, and a failed validation step is inconclusive. We do not claim a proven approximation guarantee, a polynomial-time quantum algorithm for D0LII, or a quantum speedup over exact classical MIS, SAT, or MILP solvers. A key observation that remains useful for quantum formulations is that the clique structure of $G_\theta$ makes a global size-penalty term unnecessary: the standard unconstrained MIS objective $x^\top Qx$ already targets solutions of size $k$ when a compatible D0L-system exists. This yields $\mathcal{O}(|E(G_\theta)|)$ logical two-qubit interactions per QAOA cost layer, but this is a gate count rather than a hardware depth estimate.

The experimental section supports this interpretation. The exact classical route recovers compatible systems for all $22$ positive benchmark traces and rejects all $8$ negative controls. The experiments also show that known fixed-identity symbols should be incorporated explicitly: for turtle-style systems, treating drawing symbols as unknown can create spurious compatible systems. The gate/depth table separates logical two-qubit interaction count from schedulable depth and exhibits the clique-induced high degree of the characteristic graphs. The small statevector QAOA simulations validate the implementation pathway on toy graphs, but they remain optimizer-dependent diagnostics and not evidence of practical advantage. This supports the framing: QAOA is a proof-of-concept heuristic, while the exact graph/SAT reduction is the substantive contribution.

This work opens new avenues for solving other inference problems using graph algorithms. By removing the first condition of the characteristic graph, one could model a 0L-system inference problem. Another application is that one may use adiabatic quantum computing or quantum annealing instead of QAOA. Rydberg-atom arrays provide a native physical route to MIS on unit-disk-type interaction graphs, but $G_\theta$ is generally clique-dense and we do not claim it is naturally unit-disk embeddable; any Rydberg implementation would require a separate embedding or gadget study.~\cite{EbadiEtAl2022,AndristEtAl2023,ClarkColbournJohnson1990,BreuKirkpatrick1998} Because $G_\theta$ is partitioned into local cliques, constraint-preserving mixers in the Quantum Alternating Operator Ansatz are a natural future direction; we do not study such mixers here.~\cite{HadfieldEtAl2019,FuchsEtAl2022} We also leave image-to-string front-ends, noisy observations, stochastic extensions, and model-selection criteria among multiple compatible systems for future work.

\section*{Declarations}

\paragraph{Author contributions.}
All authors contributed to the conception, design, and analysis of the work and to writing and revising the manuscript. All authors read and approved the final manuscript.

\paragraph{Funding information.}
This work was supported by the Natural Sciences and Engineering Research Council of Canada (NSERC) through Discovery Grant 2022-05092 (Ian McQuillan) and Discovery Grant RGPIN-2017-04520 (Steven Rayan).

\paragraph{Data availability.}
The implementation of the classical and quantum D0LII solvers described in this paper, together with the scripts used to generate the benchmark traces and the experimental results reported in Section~\ref{sec:experiments}, is openly available in the GitHub repository D0L-Quant-and-Classical-Solver at \url{https://github.com/alilotfi90/D0L-Quant-and-Classical-Solver}.

\paragraph{Competing interests.}
The authors have no competing interests to declare that are relevant to the content of this article.

\bibliographystyle{plain}
\bibliography{D0L_Inference_References}

@article{lotfi2024optimal,
  title={Optimal {L}-Systems for Stochastic {L}-system Inference Problems},
  author={Lotfi, Ali and McQuillan, Ian},
  journal={arXiv preprint arXiv:2409.02259},
  year={2024}
}

@book{prusinkiewicz2012algorithmic,
	title={The {A}lgorithmic {B}eauty of {P}lants},
	author={Prusinkiewicz, Przemyslaw and Lindenmayer, Aristid},
	year={2012},
	publisher={{S}pringer {S}cience \& {B}usiness {M}edia},
	address={New York}
}

@article{lindenmayer1968mathematical,
	title={Mathematical models for cellular interactions in development {I}. Filaments with one-sided inputs},
	author={Lindenmayer, Aristid},
	journal={Journal of {T}heoretical {B}iology},
	volume={18},
	number={3},
	pages={280--299},
	year={1968},
	publisher={Elsevier}
}

@unpublished{algorithmicbotany,
	author = {{University of Calgary}},
	title = {Algorithmic {B}otany},
        note = {\url{https://algorithmicbotany.org/}}
}

@article{yu2021quantum,
	title={Quantum algorithm for approximating maximum independent sets},
	author={Yu, Hongye and Wilczek, Frank and Wu, Biao},
	journal={Chinese Physics {L}etters},
	volume={38},
	number={3},
	pages={030304},
	year={2021},
	publisher={IOP Publishing}
}

@article{dorn2005quantum,
	title={Quantum complexity bounds for independent set problems},
	author={D{\"o}rn, Sebastian},
	journal={arXiv preprint quant-ph/0510084},
	year={2005}
}

@book{garey1979computers,
  title={Computers and {I}ntractability},
  author={Garey, Michael R and Johnson, David S},
  volume={174},
  year={1979},
  publisher={W.H. Freeman},
  address={San Francisco}
}

@article{farhi2014quantum,
  title={A quantum approximate optimization algorithm},
  author={Farhi, Edward and Goldstone, Jeffrey and Gutmann, Sam},
  journal={arXiv preprint arXiv:1411.4028},
  year={2014}
}

@article{ruan2023quantum,
  title={Quantum approximate optimization for combinatorial problems with constraints},
  author={Ruan, Yue and Yuan, Zhiqiang and Xue, Xiling and Liu, Zhihao},
  journal={Information {S}ciences},
  volume={619},
  pages={98--125},
  year={2023},
  publisher={Elsevier}
}

@inproceedings{choi2019tutorial,
  title={A tutorial on quantum approximate optimization algorithm ({QAOA}): Fundamentals and applications},
  author={Choi, Jaeho and Kim, Joongheon},
  booktitle={2019 {I}nternational {C}onference on {I}nformation and {C}ommunication {T}echnology {C}onvergence (ICTC)},
  pages={138--142},
  year={2019},
  organization={IEEE}
}

@book{hopcroft2001introduction,
  title={Introduction to Automata Theory, Languages, and Computation},
  author={Hopcroft, John E. and Motwani, Rajeev and Ullman, Jeffrey D.},
  year={2001},
  publisher={Addison-Wesley},
  address={Boston}
}

@book{nielsen2001quantum,
  title={Quantum {C}omputation and {Q}uantum {I}nformation},
  author={Nielsen, Michael A and Chuang, Isaac L},
  volume={2},
  year={2001},
  publisher={Cambridge University Press},
  address={Cambridge}
}

@article{cieslak2022system,
  title={{L}-system models for image-based phenomics: case studies of maize and canola},
  author={Cieslak, Mikolaj and Khan, Nazifa and Ferraro, Pascal and Soolanayakanahally, Raju and Robinson, Stephen J and Parkin, Isobel and McQuillan, Ian and Prusinkiewicz, Przemyslaw},
  journal={in silico Plants},
  volume={4},
  number={1},
  pages={diab039},
  year={2022},
  publisher={Oxford University Press UK}
}

@article{david2020global,
  title={Global wheat head detection ({GWHD}) dataset: A large and diverse dataset of high-resolution {RGB}-labelled images to develop and benchmark wheat head detection methods},
  author={David, Etienne and Madec, Simon and Sadeghi-Tehran, Pouria and Aasen, Helge and Zheng, Bangyou and Liu, Shouyang and Kirchgessner, Norbert and Ishikawa, Goro and Nagasawa, Koichi and Badhon, Minhajul A and others},
  journal={Plant Phenomics},
  year={2020},
  publisher={AAAS},
  volume={2020},
  pages={3521852},
  doi={10.34133/2020/3521852},
  note={\url{https://doi.org/10.34133/2020/3521852}}
}

@article{bernard2023stochastic,
  title={Stochastic {L}-system inference from multiple string sequence inputs},
  author={Bernard, Jason and McQuillan, Ian},
  journal={Soft {C}omputing},
  volume={27},
  number={10},
  pages={6783--6798},
  year={2023},
  publisher={Springer}
}

@article{blekos2024review,
  title={A review on quantum approximate optimization algorithm and its variants},
  author={Blekos, Kostas and Brand, Dean and Ceschini, Andrea and Chou, Chiao-Hui and Li, Rui-Hao and Pandya, Komal and Summer, Alessandro},
  journal={Physics Reports},
  volume={1068},
  pages={1--66},
  year={2024},
  publisher={Elsevier}
}

@article{wierichs2022general,
  author    = {Wierichs, David and Izaac, Josh and Wang, Cody and Lin, Cedric Yen-Yu},
  title     = {General parameter-shift rules for quantum gradients},
  journal   = {Quantum},
  volume    = {6},
  pages     = {677},
  year      = {2022},
  doi       = {10.22331/q-2022-03-30-677}
}

@inproceedings{McQuillanBernardPrusinkiewicz2018,
  author    = {Ian McQuillan and Jason Bernard and Przemyslaw Prusinkiewicz},
  title     = {Algorithms for Inferring Context-Sensitive {L}-Systems},
  booktitle = {Unconventional Computation and Natural Computation},
  series    = {Lecture Notes in Computer Science},
  volume    = {10867},
  pages     = {117--130},
  publisher = {Springer},
  year      = {2018},
  doi       = {10.1007/978-3-319-92435-9_9}
}

@article{DuffyHillisKhanMcQuillanShan2025,
  author  = {Christopher Duffy and Sam Hillis and Umer Khan and Ian McQuillan and Sonja Linghui Shan},
  title   = {Inductive inference of {Lindenmayer} systems: algorithms and computational complexity},
  journal = {Natural Computing},
  volume  = {24},
  number  = {3},
  pages   = {591--601},
  year    = {2025},
  doi     = {10.1007/s11047-025-10024-x}
}

@article{BernardMcQuillan2021,
  author  = {Jason Bernard and Ian McQuillan},
  title   = {Techniques for inferring context-free {Lindenmayer} systems with genetic algorithm},
  journal = {Swarm and Evolutionary Computation},
  volume  = {64},
  pages   = {100893},
  year    = {2021},
  doi     = {10.1016/j.swevo.2021.100893}
}

@article{Lucas2014,
  author  = {Andrew Lucas},
  title   = {Ising formulations of many {NP} problems},
  journal = {Frontiers in Physics},
  volume  = {2},
  pages   = {5},
  year    = {2014},
  doi     = {10.3389/fphy.2014.00005}
}

@article{GloverKochenbergerHennigDu2022,
  author  = {Fred Glover and Gary Kochenberger and Rick Hennig and Yu Du},
  title   = {Quantum bridge analytics {I}: a tutorial on formulating and using {QUBO} models},
  journal = {Annals of Operations Research},
  volume  = {314},
  number  = {1},
  pages   = {141--183},
  year    = {2022},
  doi     = {10.1007/s10479-022-04634-2}
}

@inproceedings{Jegou1993,
  author    = {Philippe J{\'e}gou},
  title     = {Decomposition of Domains Based on the Micro-Structure of Finite Constraint-Satisfaction Problems},
  booktitle = {Proceedings of the Eleventh National Conference on Artificial Intelligence},
  pages     = {731--736},
  year      = {1993},
  publisher = {AAAI Press}
}

@article{HadfieldEtAl2019,
  author  = {Stuart Hadfield and Zhihui Wang and Bryan O'Gorman and Eleanor G. Rieffel and Davide Venturelli and Rupak Biswas},
  title   = {From the Quantum Approximate Optimization Algorithm to a Quantum Alternating Operator Ansatz},
  journal = {Algorithms},
  volume  = {12},
  number  = {2},
  pages   = {34},
  year    = {2019},
  doi     = {10.3390/a12020034}
}

@article{FuchsEtAl2022,
  author  = {Franz G. Fuchs and Kjetil Olsen Lye and Halvor M{\o}ll Nilsen and Alexander J. Stasik and Giorgio Sartor},
  title   = {Constraint Preserving Mixers for {QAOA}},
  journal = {Algorithms},
  volume  = {15},
  number  = {6},
  pages   = {202},
  year    = {2022},
  doi     = {10.3390/a15060202}
}

@article{EbadiEtAl2022,
  author  = {Sepehr Ebadi and others},
  title   = {Quantum optimization of maximum independent set using {Rydberg} atom arrays},
  journal = {Science},
  volume  = {376},
  number  = {6598},
  pages   = {1209--1215},
  year    = {2022},
  doi     = {10.1126/science.abo6587}
}

@article{AndristEtAl2023,
  author  = {Ruben S. Andrist and Martin J. A. Schuetz and Pierre Minssen and Romina Yalovetzky and Shouvanik Chakrabarti and Dylan Herman and Niraj Kumar and Grant Salton and Ruslan Shaydulin and Yue Sun and Marco Pistoia and Helmut G. Katzgraber},
  title   = {Hardness of the maximum-independent-set problem on unit-disk graphs and prospects for quantum speedups},
  journal = {Physical Review Research},
  volume  = {5},
  pages   = {043277},
  year    = {2023},
  doi     = {10.1103/PhysRevResearch.5.043277}
}

@article{ClarkColbournJohnson1990,
  author  = {Brent N. Clark and Charles J. Colbourn and David S. Johnson},
  title   = {Unit disk graphs},
  journal = {Discrete Mathematics},
  volume  = {86},
  number  = {1--3},
  pages   = {165--177},
  year    = {1990},
  doi     = {10.1016/0012-365X(90)90358-O}
}

@article{BreuKirkpatrick1998,
  author  = {Heinz Breu and David G. Kirkpatrick},
  title   = {Unit disk graph recognition is {NP}-hard},
  journal = {Computational Geometry},
  volume  = {9},
  number  = {1--2},
  pages   = {3--24},
  year    = {1998},
  doi     = {10.1016/S0925-7721(97)00014-X}
}

\end{document}